\documentclass[11pt,a4paper,dvipsnames,usenames]{article}
\usepackage{cite}

\usepackage{amssymb}
\usepackage{amsmath}
\usepackage{amsthm}
\usepackage{xfrac}
\usepackage{latexsym}
\usepackage[dvips]{epsfig}
\usepackage{mathrsfs}
\usepackage{bm}
\usepackage{stmaryrd}
\usepackage{authblk}

\usepackage[dvipsnames]{xcolor}
\usepackage{tikz}
\usepackage{enumitem}
\usepackage{newtxtext}
\usepackage{array}

\theoremstyle{plain}
\newtheorem{proposition}{Proposition}
\newtheorem{lemma}{Lemma}
\newtheorem{theorem}{Theorem}

\newtheorem{remark}{Remark}

\allowdisplaybreaks

\def\bmd{{\bm d}}

\def\bmg{{\bm g}}

\def\bmm{{\bm m}}

\def\bmeta{{\bm \eta}}

\def\bmsigma{{\bm \sigma}}

\def\bmpartial{{\bm \partial}}

\usepackage{mathtools}
\usepackage{xparse}
\makeatletter
\newcommand{\raisemath}[1]{\mathpalette{\raisem@th{#1}}}
\newcommand{\raisem@th}[3]{\raisebox{#1}{$#2#3$}}
\makeatother
\NewDocumentCommand{\newrbar}{O{0pt} O{0pt}}{
  \ensuremath{\mathrlap{\raisemath{#2}{\hspace*{#1}{\mathchar'26\mkern-9mu}}}r}}

\usepackage{amsfonts, amscd}

\usepackage{fullpage}
\newcounter{mnotecount}

\newcommand{\mnotex}[1]
{\protect{\stepcounter{mnotecount}}$^{\mbox{\footnotesize $\bullet$\themnotecount}}$ 
\marginpar{
\raggedright\scriptsize\em
$\!\!\!\!\!\!\,\bullet$\themnotecount: #1} }

\newcounter{mnote}

\usepackage{color}
\usepackage{xcolor}
\usepackage[normalem]{ulem}
\usepackage{soul}
\usepackage{hyperref}

\begin{document}

\title{\textbf{Asymptotics of spin-0 fields and conserved charges on $n$-dimensional Minkowski spaces}}

\author[2]{Edgar Gasper\'{i}n \footnote{E-mail address: {\tt edgar.gasperin@tecnico.ulisboa.pt}}}

\author[1,3]{Mariem Magdy \footnote{E-mail
    address: {\tt mariem.magdy.a.m@tecnico.ulisboa.pt}}}

\author[1,4]{Filipe C. Mena \footnote{E-mail address: {\tt filipecmena@tecnico.ulisboa.pt}}}

\affil[1]{Centro de An\'{a}lise Matem\'{a}tica, Geometria e Sistemas Din\^{a}micos, Instituto Superior T\'{e}cnico IST, Universidade de Lisboa UL, Avenida Rovisco Pais 1, 1049-001 Lisboa, Portugal.}

\affil[2]{CENTRA, Departamento de F\'{i}sica, Instituto Superior T\'{e}cnico IST, Universidade de Lisboa UL, Avenida Rovisco Pais 1, 1049-001 Lisboa, Portugal.}

\affil[3]{Erwin Schr\"{o}dinger International Institute for Mathematics and Physics, Boltzmanngasse 9 A-1090 Vienna, Austria.}

\affil[4]{Centro de Matemática, Universidade do Minho, 4710-057 Braga, Portugal.}

\maketitle
\begin{abstract}
We use conformal geometry methods and the construction of Friedrich's cylinder at spatial infinity to study the propagation of spin-$0$ fields (solutions to the wave equation) on $n$-dimensional Minkowski spacetimes in a neighbourhood of spatial and null infinity.  We obtain formal solutions written in terms of series expansions close to spatial and null infinity and use them to compute non-trivial asymptotic spin-$0$ charges. It is shown that if one examines the most general initial data within the class considered in this paper, the expansion is polyhomogeneous and hence of restricted regularity at null infinity. Furthermore, we derive the conditions on the initial data needed to obtain well-defined limits for the asymptotic charges at the critical sets where null infinity and spatial infinity meet. In even dimensions, we find that there are infinitely many well-defined asymptotic charges at the critical sets, while for odd dimensions there exist no non-trivial asymptotic charges with well-defined limits at the critical sets. 
\\\\
Keywords: Conformal geometry; Wave equation; Asymptotic expansions; Asymptotic charges

\end{abstract}
\newpage
\tableofcontents
\newpage
\section{Introduction}
The conformal Einstein field equations of Friedrich and, more generally, conformal methods constitute a powerful set of tools in General Relativity that allows the use of local methods to study the global geometric structure of spacetime \cite{JVKbook}. Those methods have been utilised mainly in the proof of the non-linear stability of certain classes of spacetimes \cite{Alho-Mena-Kroon, Fri-EMYM, Fri-dust, Joudioux, LueVal12b, LueVal10, Luebbe-Mena}. Historically, the first application of these methods was the global and semi-global non-linear stability of the de-Sitter and Minkowski spacetime in \cite{Fri-dS}, respectively. The adjective \emph{semi-global} for the Minkowski spacetime result is used to emphasise that the initial data, in this case, was not given on a Cauchy hypersurface —which terminates at spatial infinity $i^0$— but rather on a hyperboloidal hypersurface —which terminates at null infinity $\mathscr{I}$. 

Interestingly, the use of hyperboloidal hypersurfaces is becoming an increasingly popular approach in Numerical Relativity to reach null infinity even without using the conformal Einstein field equations \cite{Zen08, BarSarBuc11, VanHusHil14, HilHarBug16, VanHus17, GasHil18}. Although using hyperboloidal hypersurfaces is a valid strategy to obtain semi-global results, the main obstacle in obtaining a fully global non-linear stability result for the Minkowski spacetime —say, as that of \cite{ChrKla93, LinRod04}— using conformal methods is localised on a small neighbourhood of spatial infinity. The reason is that the conformal structure of asymptotically flat spacetimes with non-vanishing mass becomes degenerate at $i^0$ as observed since the decade of 1960 by Penrose \cite{Pen65a}. This is the so-called \emph{problem at spatial infinity}. The problem at $i^0$ has several consequences on the gravitational field such as the “violation” of the peeling theorem as evidenced by formal polyhomogeneous asymptotic expansions obtained using conformal methods in \cite{GasVal17a, Fri98a}. Nonetheless, a milestone in the resolution of the problem at $i^0$ has been the construction of extended representations of $i^0$ or “blow-ups”, for instance, Friedrich’s cylinder at spatial infinity and Ashtekar’s hyperboloid at spatial infinity \cite{Ash1,Ash2}. As the relation between these two representations of spatial infinity has been discussed \cite{MagVal21}, we will focus only on Friedrich’s representation. 

The original construction of the cylinder at $i^0$ for general asymptotically flat spacetimes can be traced back to \cite{Fri98a} where expansions of various field quantities are obtained around the region of spacetime where null infinity meets spacelike infinity. This allowed determining the behaviour of various fields at null infinity while connecting the coefficients in the expansion with the initial data for the fields’ evolution. An iterative procedure introduced in \cite{Fri98a} establishes conditions for the existence and regularity of those quantities. These conditions ensure that there are no logarithmic terms in the asymptotic solutions up to a given order \cite{Val04a, Val04d} in the non-linear case or to arbitrary order in the case of linear fields \cite{GasVal21, Val03a}.

\medskip

A particularly timely application of these methods is the calculation of asymptotic charges. The prime example is the computation of the gravitational Newman-Penrose constants in \cite{Fried-bondi}. More recently, an investigation of the Newman-Penrose constants for the spin-1, spin-2, and spin-$0$ fields in the Minkowski spacetime has been given in \cite{GasVal20} and \cite{Gasperin-Pinto}. Naturally, the study of asymptotic conserved charges has a long history in General Relativity, preceding, of course, the construction of the cylinder at $i^0$, going back at least to the works of Komar \cite{Komar} and Arnowitt, Deser, and Misner \cite{Arnowitt} who linked the existence of local symmetry transformations at infinity with conserved quantities. Later, in a seminal work, Bondi, van der Burg, Metzner \cite{BMS}, and Sachs \cite{Sachs} found an infinite number of charges associated with asymptotic symmetries which satisfy an algebra generalising the Poincaré algebra (nowadays called BMS algebra). 

\medskip 

Given the importance of the wave equation for metric formulations of General Relativity such as the harmonic gauge formulation \cite{Cho52, LinRod04} and other field theories, a natural problem to consider is the existence of asymptotically regular spin-$0$ fields and their associated asymptotic charges. There is a vast literature on the topic of asymptotic charges \cite{HennTro18,HennTro18-2,HennTro18-3,WaldZoupas2000,BarTro16,CamEyer17,CapNgyPari23,NguyenWest22,HennTroe19,NguyenWest22-2,PraShe19,Pra18,Pra19,PraShe22} with notable work \cite{NguyenWest22,NguyenWest22-2,CamCoiMiz18,CamFreHopfSoni19} focusing on the asymptotic charges associated with spin-$0$ fields. The primary distinction between previous work and our analysis is that we exploit Friedrich's cylinder at spatial infinity to study the asymptotic regularity of spin-$0$ fields given by solutions to the wave equation on $n$-dimensional Minkowski spacetimes. In this context, the conformal spin-0 fields are required to exhibit at least $C^{0}$-regularity near null infinity to ensure well-defined limits of the asymptotic charges. Nevertheless, we will see that more stringent initial data conditions allow our solutions to extend smoothly to null infinity, although smooth extensions are not explored in all cases --- see Remark \ref{Non-integer-nu-regularity}. Therefore, we start with an ansatz in the form of a Taylor-like expansion of the field close to $i^0$ and use the obtained solutions to compute the asymptotic charges at the critical sets where null infinity and spatial infinity meet. A similar strategy was used to study the asymptotic charges associated with spin-$1$ and spin-$2$ fields propagating close to spatial and null infinity on a Minkowski background \cite{MagdyKroon22}. More recently, this was generalised to the non-linear gravitational field \cite{MagVal24}. This particular strategy restricts the initial data to be analytic, and it is shown that general initial data within the class considered here will lead to polyhomogeneous expansions. These can be interpreted as the linear spin-$0$ counterpart of the peeling violations of \cite{GasVal17a}. 

Due to the use of an ansatz, our solution is only \emph{formal} in the sense that we do not make a full PDE analysis to address the convergence of the solution expressed as an infinite series expansion — although each partial sum is an \emph{exact solution} for some initial data. Issues of convergence and PDE energy estimates will be addressed elsewhere. Despite these caveats, the approach presents several advantages: it provides a framework where it is easier to obtain the asymptotic solutions to the wave equation, it allows a direct connection between those solutions and the initial data, and, in this way, it allows obtaining a global picture of the solution. Moreover, our conformal solutions can be related to the physical solutions of the wave equation and the logarithmic terms that appear in our setting can be seen in some examples. In this way, the terms that may compromise the conformal regularity are controlled directly in terms of the initial data.
\medskip

In particular, we find initial data conditions that implies asymptotic solutions to the wave equation that extend smoothly to null infinity on $n$-dimensional Minkowski backgrounds. In the 4-dimensional case, we find that there are infinitely many asymptotic conserved charges associated with the spin-0 field. For higher dimensions, we find only a finite number of non-trivial asymptotic solutions with well-defined limits at null infinity. The existence of these \emph{non-trivial} charges crucially depends on the existence of integer solutions to an algebraic equation. We note that our results are independent from the Einstein field equations of General Relativity in the sense that the background is flat and hence it can be considered as a simple model for any field theory whose evolution equations are wave equations and admit some notion of asymptotic flatness compatible with the construction of the cylinder at $i^0$.

 \medskip
 
The paper is organised as follows: In Section \ref{conformal}, we
reconstruct Friedrich's cylinder on the $n$-dimensional Minkowski
spacetime. In Section \ref{asymp-sec}, we establish conditions for the
existence of $C^{s}$-regular ($s \geq 0$) solutions to the conformal wave equation on the
Minkowski background, with supplementary material in Appendix
\ref{Appendix:Solutions-Legendre-ODE}. As those conditions depend on the spacetime dimension, this section is split into two subsections. Finally, Section \ref{charges-sec} contains an application of the results of the previous section to the computation of asymptotically conserved charges in our setting.
\section{Conformal compactification and the Friedrich cylinder}
\label{conformal}
This section introduces Friedrich's cylinder on an $n$-dimensional
Minkowski space which is a generalisation of the original
four-dimensional version \cite{Fried-15}.
\subsection{Preliminaries}
Given an $n$-dimensional Lorentzian space $(\tilde{\mathcal{M}},\tilde{\bmg})$, which we call {\em physical space}, we introduce the conformal Lorentzian space $(\mathcal{M},\bmg)$, called {\em unphysical space}, so that $\bmg$
and $\tilde{\bmg}$ are related by
\begin{equation}
    \bmg = \Xi^{2} \tilde{\bmg},
    \label{Conformal-transofrmation}
\end{equation}
where $\Xi$ is a positive function commonly referred to as the
\emph{conformal factor}. This conformal transformation is induced by a diffeomorphic map $\Phi: \tilde{\mathcal{M}} \to \mathcal{M}$ such that $\Phi^{*} \bmg = \Xi^{2} \tilde{\bmg}$, where we used $\bmg$ as a shorthand notation for $\Phi^{*} \bmg$ in
Eq. \eqref{Conformal-transofrmation}. The conformal transformation
$\Phi$ defines a conformal compactification of $\tilde{\mathcal{M}}$
if $\Xi$ satisfies:
\begin{enumerate}
        \item[i.] $\Xi > 0$ on a compact, connected and open subset
          $\mathcal{U}$ of $\mathcal{M}$,
        \item[ii.] $\Xi =0$ on the boundary $\partial \mathcal{U}$ of the open set $\mathcal{U}$.
\end{enumerate}
Then, one can say that $\partial \mathcal{U}$ is the conformal
boundary of $\tilde{\mathcal{M}}$.
The conformal map also transforms fields and operators on
$\tilde{\mathcal{M}}$ into $\mathcal{M}$ and a non-trivial problem is to prove the regularity of those fields at the conformal boundary.  In this paper, we assume that $\tilde{\mathcal{M}}$ is the $n$-dimensional Minkowski space.

We are interested in establishing regular solutions to the wave equation in a neighbourhood where past and future null conformal infinity meet with spatial infinity. To do that, we use Friedrich's description of spatial infinity in terms of a cylinder.

\subsection{Friedrich's cylinder on the $n$-dimensional Minkowski space}
 Consider the Minkowski metric $\tilde{\bmeta}$ in spherical
 coordinates $(\tilde{t}, \tilde{r},\tilde{\theta}^{\mathcal{A}})$
\begin{equation*}
    \tilde{\bmeta} = - \bmd{\tilde{t}} \otimes \bmd{\tilde{t}} +
    \bmd{\tilde{r}} \otimes \bmd{\tilde{r}} + \tilde{r}^2 \bmsigma,
    \qquad \tilde{t} \in (- \infty, \infty), \quad \tilde{r} \in
           [0,\infty),
\end{equation*}
where $\bmsigma$ is the standard metric on $\mathbb{S}^{n-2}$ with
some choice of spherical polar coordinates
$(\tilde{\theta}^{\mathcal{A}})$, with $\mathcal{A} =
1,2,...,n-2$. Then, we introduce the coordinate transformation
$(\tilde{t},\tilde{r},\tilde{\theta}^{\mathcal{A}}) \to (t,\rho,
\theta^{\mathcal{A}})$ given by
\begin{equation}
    t = \frac{\tilde{t}}{\tilde{r}^{2}-\tilde{t}^2}, \qquad \rho =
    \frac{\tilde{r}}{\tilde{r}^2-\tilde{t}^2}, \qquad
    \theta^{\mathcal{A}} = \tilde{\theta}^{\mathcal{A}},
    \label{Coordinate-transformation}
\end{equation}
so that the spatial infinity point $i^{0} (\tilde{r}\to \infty)$
correspond to the origin $(t=0,\rho=0)$ in the new coordinate
system. In terms of $(t,\rho, \theta^{\mathcal{A}})$, we can write
$\tilde{\bmeta}$ as
\begin{equation*}
  \tilde{\bmeta} = \frac{1}{(\rho^2-t^2)^2} \left( -\bmd{t} \otimes
  \bmd{t} + \bmd{\rho}\otimes \bmd{\rho} + \rho^2 \bmsigma \right),
  \qquad t \in (-\infty, \infty), \qquad \rho \in [0,\infty).
\end{equation*}
Note that the above defines the conformal metric $\hat{\bmeta}= \Xi^2
\tilde{\bmeta}$ with $\Xi = \rho^2 - t^2$ so that
\begin{equation*}
    \hat{\bmeta} = -\bmd{t} \otimes \bmd{t} + \bmd{\rho}\otimes
    \bmd{\rho} + \rho^2 \bmsigma.
\end{equation*}
Let us introduce $\tau = t/ \rho$ and define the conformal factor
\begin{equation*}
    \Theta \equiv \frac{\Xi}{\rho} = \rho (1-\tau^2),
\end{equation*}
so that the conformal metric
\begin{equation}
    \bmeta = \Theta^2 \tilde{\bmeta}
    \label{physical-to-unphysical-metric}
\end{equation}
is given by
\begin{equation}
  \bmeta = - \bmd{\tau} \otimes \bmd{\tau} + \frac{(1-\tau^2)}{\rho^2}
  \bmd{\rho} \otimes \bmd{\rho} - \frac{\tau}{\rho} \left( \bmd{\tau}
  \otimes \bmd{\rho} + \bmd{\rho} \otimes \bmd{\tau} \right) +
  \bmsigma.
    \label{Friedrich-metric}
\end{equation}
Given the above, consider the conformal extension $(\mathcal{M},
\bmeta)$ where $\bmeta$ is given by Eq. \eqref{Friedrich-metric} and
\begin{equation*}
    \mathcal{M} = \{ p \in \mathbb{R}^n \ | \ -1 \leq \tau(p)
    \leq 1, \ \rho(p) \geq 0 \}.
\end{equation*}
Then introduce the following subsets of the conformal boundary
$\Theta=0$:
\begin{subequations}
    \begin{align*}
        \mathscr{I}^\pm & =\big\{ p\in  \mathcal{M} \hspace{1mm}\rvert \hspace{1mm} \tau(p) =\pm 1
        \big\}, \qquad & \text{past and future null infinity}
        \\ \mathcal{I} & = \big\{ p \in \mathcal{M} \hspace{1mm}  \rvert \hspace{1mm} |\tau(p)|<1, \hspace{1mm} \rho(p)=0\big\},
        \qquad & \text{the cylinder at spatial infinity}
        \\ \mathcal{I}^{\pm} & = \big\{ p\in  \mathcal{M} \hspace{1mm} \rvert \hspace{1mm} \tau(p)= \pm
        1, \hspace{1mm} \rho(p)=0 \big\}, \qquad & \text{the critical    sets of null infinity}
    \end{align*}
\end{subequations}
and 
\begin{equation*}
    \mathcal{I}^{0} \equiv \big\{ p \in \mathcal{M}\hspace{1mm} \rvert \hspace{1mm} \tau(p)=0, \hspace{1mm} \rho(p)=0\big\}
\end{equation*}
denoting the intersection of $\mathcal{I}$ with the initial hypersurface 
\begin{equation}
\mathcal{S}_{\star} = \{  p\in  \mathcal{M} \hspace{1mm}\rvert  \hspace{1mm} \tau(p) =0 \}.
\end{equation}
Notice that this corresponds to the physical Cauchy hypersurface $\tilde{t}=0$
for any $\tilde{r} \neq 0$.

In this conformal representation, the spatial infinity point $i^{0}$ corresponds to the set $\mathcal{I}$ with the topology $\mathbb{R} \times \mathbb{S}^{n-2}$, hence the name \emph{cylinder at spatial infinity}. The construction of the cylinder at spatial infinity is then based on the blow-up of a point $q\in \mathcal{S}_\star$  to a codimension-$1$ sphere. So given that $\mathcal{S}_\star$ is an $n-1$ Riemannian manifold, then the blow-up of $q$ is a $n-2$ dimensional sphere.

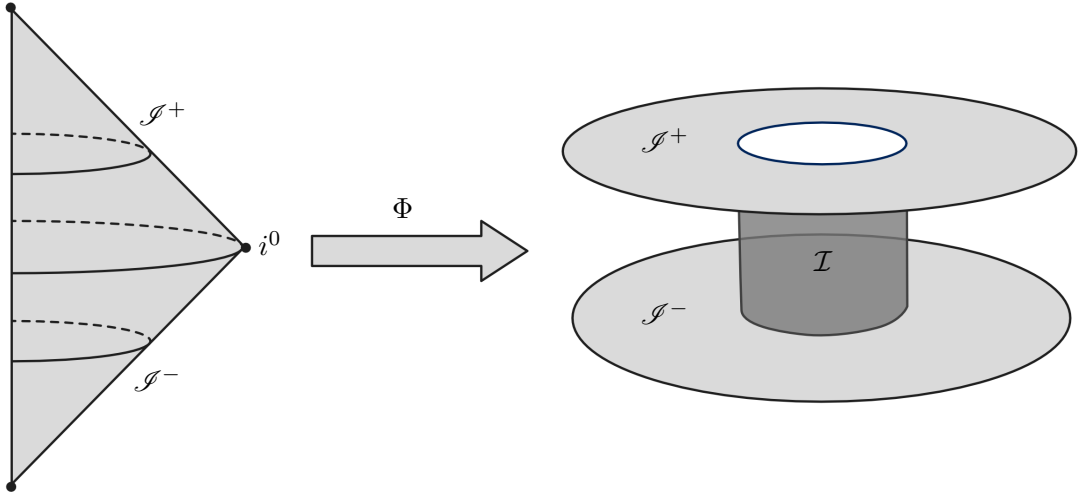
\begin{figure}[t]
\centering
\includegraphics[width=150mm]{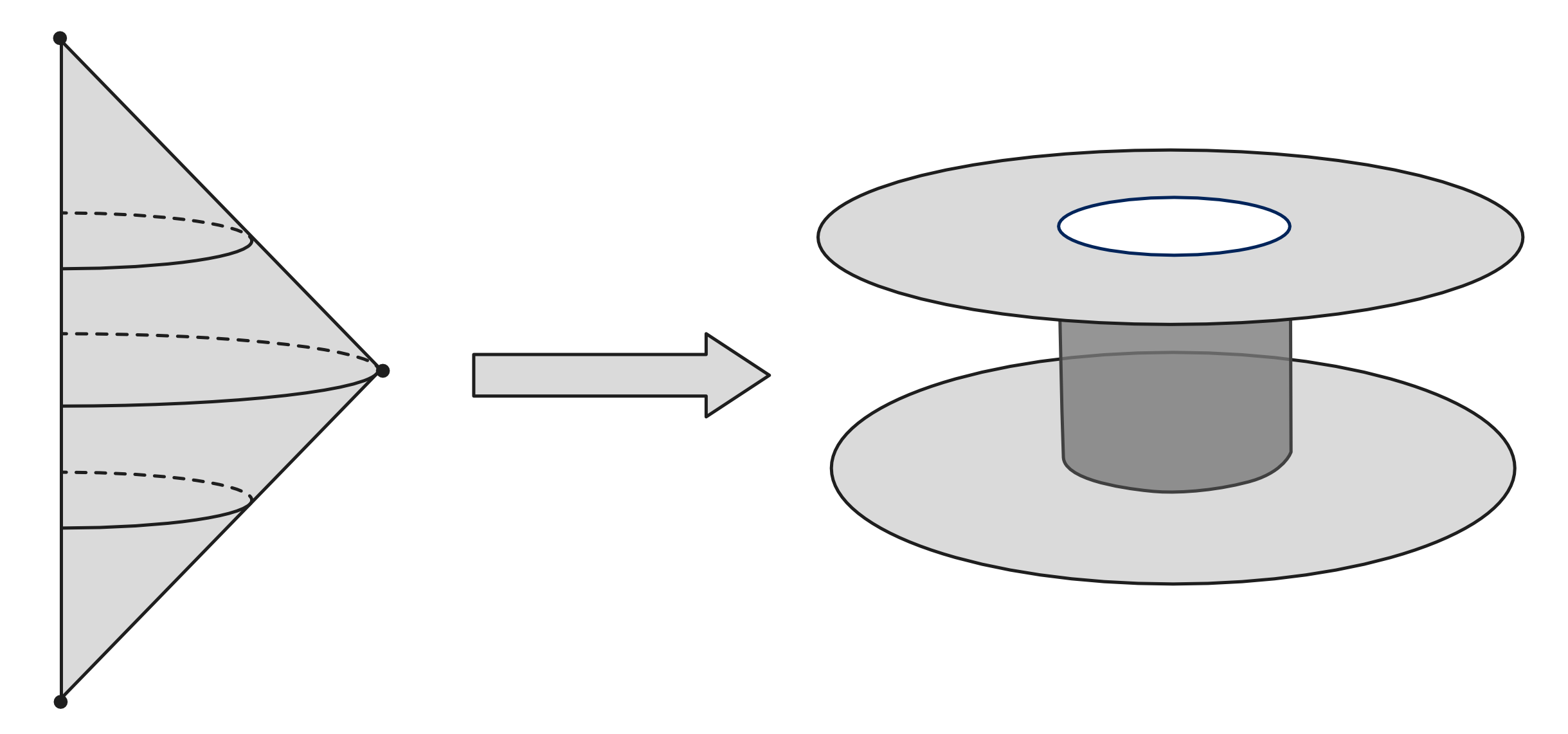}
\put(-175,135){$\mathscr{I}^+$}
\put(-110,90){$\mathcal{I}$}
\put(-268,110){$\Phi$}
\put(-175,70){$\mathscr{I}^-$}
\put(-363,145){$\mathscr{I}^+$}
\put(-318,96){$i^{0}$}
\put(-365,45){$\mathscr{I}^-$}
\caption{A schematic representation of the conformal transformation map $\Phi: \hat{\bmeta} \to \bmeta$. In particular, the diagram illustrates the blow-up of $i^{0} \cong \mathbb{S}^{n-2}$ to Friedrich's cylinder $\mathcal{I} \cong \mathbb{R} \times \mathbb{S}^{n-2}$ that connects $\mathscr{I}^{+}$ and $\mathscr{I}^{-}$.}
\label{Fig:point-compactification-to-Friedrich-cylinder}
\end{figure}

\section{Asymptotic solutions to the wave equation in $n$-dimensions}
\label{asymp-sec}
The evolution of a spin-$0$ field on Minkowski spacetime is governed by a scalar field satisfying the wave equation
\begin{equation}
    \tilde{\square}_{\tilde\bmeta} \tilde{\phi} = \frac{1}{\sqrt{- \text{det} \tilde{\bmeta}}} \bmpartial_{\mu} (\sqrt{- \text{det} \tilde{\bmeta}} \bmpartial^{\mu} \phi) = 0,
    \label{spin-$0$-field-equation}
\end{equation}
where $\tilde{\square}_{\tilde\bmeta}  = \tilde{\eta}^{ab} \tilde{\nabla}_{a} \tilde{\nabla}_{b}$ is the wave operator relative to the metric  $\tilde{\bmeta}$. 

It is known that given Eq.~\eqref{physical-to-unphysical-metric}, the conformal transformation of the wave operator $\tilde{\square}_{\tilde\bmeta}$ implies
\begin{equation}
   \left(g^{\mu\nu}\nabla_\mu\nabla_\nu-\frac{n-2}{4(n-1)} R\right)\left( \Theta^{1-n/2} \tilde\phi\right)=\Theta^{-1-n/2}\left( \tilde g^{\mu\nu}\tilde\nabla_\mu\tilde\nabla_\nu-\frac{n-2}{4(n-1)} \tilde R  \right)\tilde\phi,
\end{equation}
and for $\tilde R=0$, this simply gives
\begin{equation}
\square_\bmeta \phi -\frac{n-2}{4(n-1)} R \phi =0,
 \label{unphysical-spin-$0$-equation}
\end{equation}
where $\phi = \Theta^{1-n/2} \tilde{\phi}$ and $R = (n-4) (n-1)$.

In terms of Friedrich's coordinates $(\tau, \rho, \theta^{\mathcal{A}})$, Eq. \eqref{unphysical-spin-$0$-equation} can be written as
\begin{equation}
    (1-\tau^2) \bmpartial_{\tau}^{2}\phi  + 2 \rho \tau \bmpartial_{\rho} \bmpartial_{\tau} \phi - 2 \tau \bmpartial_{\tau} \phi - \rho^2 \bmpartial_{\rho}^{2} \phi - \Delta_{\mathbb{S}^{n-2}} \phi  -\frac{n-2}{4(n-1)} R \phi=0,
    \label{unphysical-spin-$0$-equation-expanded}
\end{equation}
where $\Delta_{\mathbb{S}^{n-2}}$ is the Laplacian operator on $\mathbb{S}^{n-2}$. This is the main equation we will analyse in this work. 

Following the analysis of the spin-$0$ field equation in \cite{Gasperin-Pinto} and similar analyses \cite{MinucciPanossoKroon22,MagdyKroon22} in four-dimensions, we consider the following ansatz
\begin{equation}
    \phi (\tau,\rho,\theta^A)= \sum_{p=0}^{\infty} \sum_{\ell= 0}^{\infty}  \sum_{\bmm}  \frac{1}{p!} a_{p; \ell,\bmm }(\tau) Y_{\ell,\bmm} \rho^{p},
    \label{Expansion-phi}
\end{equation}
where $Y_{\ell,\bmm}=Y_{\ell,\bmm}(\theta^{\mathcal{A}})$ are the standard spherical harmonics on $\mathbb{S}^{n-2}$ with multi-index $\bmm=(\ell_1,\ell_2,\ldots,\ell_{n-3}$) using the notation of \cite{Avery}. Notice here that this ansatz is more general than the one considered in \cite{Gasperin-Pinto, MinucciPanossoKroon22} since $\ell$ is allowed to be greater than $p$. As we shall see in the sequel, this is required to obtain non-trivial asymptotic charges.

Consistent with this expansion, the initial data for $\phi$ on $\mathcal{S}_{\star}$ can be written as
\begin{equation}
\aligned
\label{initial-data}
&& \phi|_{\mathcal{S}_{\star}} (\rho,\theta^A)=\sum_{p=0}^{\infty} \sum_{\ell= 0}^{\infty}  \sum_{\bmm}  \frac{1}{p!} a_{p;\ell,\bmm}(0) Y_{\ell,\bmm} \rho^{p} \\
&& \dot{\phi}|_{\mathcal{S}_{\star}} (\rho,\theta^A)= \sum_{p=0}^{\infty} \sum_{\ell= 0}^{\infty}  \sum_{\bmm}  \frac{1}{p!} \dot a_{p;\ell,\bmm}(0) Y_{\ell,\bmm} \rho^{p}.
\endaligned
\end{equation}
By substituting Eq.~\eqref{Expansion-phi} into Eq.~\eqref{unphysical-spin-$0$-equation-expanded}, we get the following equation for $p>0$ 
\begin{equation}
    (1-\tau^2) \ddot{a}_{p;\ell,\bmm} + 2 \tau (p-1) \dot{a}_{p; \ell,\bmm} + \frac{1}{4} (2p+n+2l-4) (2l+n-2(1+p)) a_{p;\ell,\bmm} =0.
    \label{ODE-general-a-coefficient}
\end{equation}
where a dot denotes differentiation with respect to $\tau$.
\begin{lemma}
\label{lemma1}
    Given initial data as in Eq.~\eqref{initial-data} at $\tau=0$, the solution to Eq.~\eqref{ODE-general-a-coefficient} for any $p \geq 0$ is given by
        \begin{equation}
            \label{solution-exp}
            a_{p;\ell,\bmm}(\tau) = \left( C_{p;\ell,\bmm} P^{p}_{\nu}(\tau) + D_{p;\ell,\bmm} Q^{p}_{\nu}(\tau) \right)(1-\tau^2)^{\tfrac{p}{2}},
        \end{equation}
        where $P^{p}_{\nu}(\tau)$ is the associated Legendre polynomial of order $\nu$, $Q^{p}_{\nu}(\tau)$ is the associated Legendre function of the second kind of order $\nu$ and 
        \begin{equation}
        \label{nu-parameter}
        \nu=\frac{1}{2} (n+2l-4).
        \end{equation}
        The constants $C_{p;\ell,\bmm}$ and $D_{p;\ell,\bmm}$ are given by
        \begin{subequations}
            \begin{eqnarray}
                && C_{p;\ell,\bmm} = -\frac{(-2+2l+n-2p) a_{p;\ell,\bmm}(0) Q^{p}_{\nu+1}(0) -2 \dot{a}_{p;\ell,\bmm}(0) Q^{p}_{\nu}(0)}{
                (2l+n-2(1+p))(P^{p}_{\nu}(0)Q^{p}_{\nu+1}(0)-P^{p}_{\nu+1}(0)Q^{p}_{\nu}(0))}, \label{Constant-Cplm} \\
                && D_{p;\ell,\bmm} = -\frac{ (-2+2l+n-2p) a_{p;\ell,\bmm}(0) P^{p}_{\nu}(0) -2 \dot{a}_{p;\ell,\bmm}(0) P^{p}_{\nu+1}(0) }{
                (2l+n-2(1+p))(P^{p}_{\nu}(0)Q^{p}_{\nu+1}(0)-P^{p}_{\nu+1}(0)Q^{p}_{\nu}(0))}. \label{Constant-Dplm}
            \end{eqnarray}
        \end{subequations}
     \label{Solution-a}
\end{lemma}
\begin{proof} 
    Starting from Eq. \eqref{ODE-general-a-coefficient} and making the substitution 
    \begin{equation*}
        a_{p;\ell,\bmm}(\tau) = (1-\tau^2)^{\tfrac{p}{2}} w(\tau),
    \end{equation*}
we get that Eq.~\eqref{ODE-general-a-coefficient} is a Legendre equation of the form
\begin{equation*}
    (1-\tau^2) \ddot{w}(\tau) - 2 \tau \dot{w}(\tau) + \left( l^2+l(n-3) + \frac{1}{4} \left(n^2 - 6n - \frac{4(-2+p^2+2\tau^2)}{1-\tau^2} \right)
    \right) w(\tau) =0.
\end{equation*}
Comparing with Eq. \eqref{differential-Equation-Appendix-A}, we see that  $\mu = \pm p$ and $\nu = \nu_{1} = \tfrac{1}{2} (2-2l-n)$ or $\nu = \nu_{2} = \tfrac{1}{2} (-4+2l+n)$. However, given that $P^{\pm \mu}_{\nu}(\pm z), Q^{\pm \mu}_{\nu}(\pm z), P^{\pm \mu}_{-\nu -1}(\pm z)$ and $Q^{\pm \mu}_{-\nu -1}(\pm z)$ are also solutions to Eq. \eqref{differential-Equation-Appendix-A} since Eq. \eqref{differential-Equation-Appendix-A} is invariant under the replacements $\mu \to - \mu$, $z \to - z$, $\nu \to - \nu - 1$, we can, without loss of generality, restrict to $\mu=p$ and $\nu = \nu_{2}$. Therefore, the solution to Eq. \eqref{ODE-general-a-coefficient} can be written in the form of Eq.~\eqref{solution-exp}. See Appendix \ref{Appendix:Solutions-Legendre-ODE} for more details on Legendre differential equations. 
\end{proof}
The asymptotic $C^{s}$ ($s \geq 0$) regularity of $a_{p;\ell,\bmm}(\tau)$ near $\tau = \pm 1$ crucially depends on whether $\nu$ is an integer or not, so next we split our analysis accordingly.
\subsection{Solutions for  $\nu \in \mathbb{Z}$}
By inspecting Eq.~\eqref{nu-parameter}, we see that $\nu \in \mathbb{Z}$ for any $l\ge 0$ for any even $n$ and it is a non-integer for any $l\ge 0$ for any odd $n$. For $n=4$, we have $\nu= l$. Without loss of generality, we restrict the analysis in the rest of this section to $\nu \in \mathbb{Z}^{+}$. Then, one can show the following:
\begin{proposition} \label{prop:Regularity-conditions-integer-nu}
    Consider $\nu \in \mathbb{Z}^{+}$. Then
    \begin{enumerate}
    \item For $\nu < p$, the solution given by Eq.~\eqref{solution-exp} in Lemma \ref{lemma1} is always smooth near $\tau = \pm 1$. 
    \item For  $\nu \ge p$, the solution given by Eq.~\eqref{solution-exp} in Lemma \ref{Solution-a} is smooth near $\tau = \pm 1$ if and only if,
    \begin{enumerate}
        \item [(i)] $a_{p;\ell,\bmm}(0) =0$  for odd $\nu+p$, and
        \item [(ii)] $\dot{a}_{p;\ell,\bmm}(0)=0$  for even $\nu+p$.
    \end{enumerate}
    \end{enumerate}
    \label{Regularity-conditions-integer-nu}
\end{proposition}
\begin{proof}
    $(i)$  For $\nu < p$, $P^{p}_{\nu}(\tau)=0$ since $P^{p}_{\nu}(\tau)$ is given by 
    \begin{eqnarray}
    && P^{p}_{\nu}(\tau)= \frac{(-1)^{p+\nu}}{2^{\nu} \nu!} (1-\tau^2)^{\tfrac{p}{2}} \frac{d^{p+\nu} }{d \tau^{p+\nu}} \left( (1-\tau)^{\nu} (1+\tau)^{\nu} \right), \label{Associated-legendre-polynomial-integer-nu-expanded} \\
    && \phantom{ P^{p}_{\nu}(\tau)} = \sum_{i=p}^{\nu} \frac{(-1)^{i+p+\nu}}{2^{\nu} \nu!} \binom{p+\nu}{i} \nu^{(i)} \nu^{(p-i+\nu)} (1-\tau)^{\nu-i+\tfrac{p}{2}} (1+\tau)^{i-\tfrac{p}{2}}, \nonumber
    \end{eqnarray}
    where $x^{(i)} := x (x-1) \ldots (x-(i-1))$ and in the above we used Eq. \eqref{Definition-Legendre-polynomial} and Eq. \eqref{Definition-associated-legendre-polynomial} and the fact that for arbitrary functions $f(\tau)$ and $g(\tau)$:
        \begin{equation}
            \frac{d^{s}}{d\tau^{s}} (f(\tau) g(\tau)) = \sum_{i=0}^{s} \binom{s}{i} f^{(i)}(\tau) g^{(s-i)}(\tau).
            \label{Derivative-two-functions}
        \end{equation}
    Therefore, for $\nu < p$, the solution $a_{p;\ell,\bmm}(\tau)$ can be written as
\begin{equation}
    a_{p;\ell,\bmm}(\tau)=  D_{p;\ell,\bmm} Q^{p}_{\nu}(\tau) (1-\tau^2)^{\tfrac{p}{2}}.
    \label{Solution-nu-less-than-p}
\end{equation}
Then, using Eq. \eqref{Definition-Legendre-function-2} and Eq. \eqref{Derivative-two-functions}, it is straightforward to show that for $\nu \in \mathbb{Z}^{+}$, we have
    \begin{eqnarray}
    && Q^{(p)}_{\nu}(\tau) =  \frac{d^{p}}{d \tau^{p}} Q_{\nu}(\tau) = \frac{1}{2} P^{(p)}_{\nu}(\tau) \log\left( \frac{1+\tau}{1-\tau} \right) \nonumber  \\
    && \phantom{Q^{(p)}_{\nu}(\tau)= \frac{d^{p}}{d \tau^{p}} Q_{\nu}(\tau)} + \frac{1}{2} \sum_{i=1}^{p} \frac{i!}{i}  \binom{p}{i} \left( (-1)^{-i+1} (1+\tau)^{-i} + (1-\tau)^{-i} \right) P_{\nu}^{(p-i)}(\tau) \label{p-derivative-Legendre-function-2} \\
    && \phantom{Q^{(p)}_{\nu}(\tau)= \frac{d^{p}}{d \tau^{p}} Q_{\nu}(\tau)} - \sum_{i=0}^{\tfrac{\nu-1}{2}} \frac{(2\nu-4i-1)}{(\nu-1)(2i+1)} P^{(p)}_{\nu-2i-1}(\tau). \nonumber
\end{eqnarray}
Eq. \eqref{Definition-associated-legendre-function-2} then implies  
\begin{equation*}
    Q^{p}_{\nu}(\tau) = \frac{(-1)^{p}}{2} (1-\tau^2)^{\tfrac{p}{2}} \sum_{i=1}^{p} \frac{i!}{i}  \binom{p}{i} \left( (-1)^{-i+1} (1+\tau)^{-i} + (1-\tau)^{-i} \right) P_{\nu}^{(p-i)}(\tau), 
\end{equation*}
where we used that for $\nu < p$, both $P^{(p)}_{\nu}(\tau)$ and $P^{(p)}_{\nu-2i-1}(\tau)$ are vanishing for any $\tau$. Substituting in Eq. \eqref{Solution-nu-less-than-p}, we get
\begin{equation*}
    a_{p;\ell,\bmm}(\tau) = \frac{1}{2} D_{p;\ell,\bmm} \sum_{i=1}^{p} \frac{i! (-1)^{p-i}}{i} \binom{p}{i} (1-\tau^2)^{p-i} \left( - (1-\tau)^{i} + (-1)^{i} (1+\tau)^{i} \right) P_{\nu}^{(p-i)}(\tau),  
    \label{solution-general-nu-less-p}
\end{equation*}
which is smooth near $\tau=\pm 1$ given that $p-i\geq0$.
\\

$(ii)$ For  $\nu \ge p$, note that the log term in Eq. \eqref{p-derivative-Legendre-function-2} will be non-vanishing, in general. Therefore, for $\nu \in \mathbb{Z}^{+}$,  $a_{p;\ell,\bmm}(\tau)$ diverges logarithmically at $\tau = \pm 1$ unless we enforce the condition $D_{p;\ell,\bmm}=0$. Then the regularity conditions follows from the observation that the coefficient of $\dot{a}_{p;\ell,\bmm}(0)$ in Eq. \eqref{Constant-Dplm} is vanishing when $\nu+p$ is odd while the coefficient of $a_{p:\ell,\bmm}(0)$ is vanishing when $\nu+p$ is even. In particular, for $\nu \in \mathbb{Z}$, we have
\begin{equation*}
    P^{p}_{\nu}(0) \propto \frac{\Gamma(\tfrac{1+p+\nu}{2})}{\Gamma(1-p+\nu) \Gamma(\tfrac{1+p-\nu}{2})},
\end{equation*}
which is vanishing only at the poles of $\Gamma(1-p+\nu)$ and $\Gamma(\tfrac{1+p-\nu}{2})$. For $\nu \geq p, \Gamma(1-p+\nu)$ is never diverging whereas $\Gamma(\tfrac{1+p-\nu}{2})$ will diverge for odd $p+\nu$ given that $\tfrac{1+p-\nu}{2} \leq \tfrac{1}{2}$ and that the only integers less than $\tfrac{1}{2}$ are exactly those where $\Gamma(\tfrac{1+p-\nu}{2})$ has poles. Similar argument follows for $P^{p}_{\nu+1}(0)$.
\end{proof}

We emphasize that since we get log terms in the solutions, i.e polyhomogeneous solutions, then $a_{p;\ell,\bmm}(\tau)$ will not be $C^{\infty}$ near $\tau= \pm 1$ unless the regularity conditions are imposed.

\subsubsection{The subcase $\nu\ge p$}
For the computation of the asymptotic conserved charges in Section \ref{charges-sec}, we will be especially interested in the case $\nu\ge p$. So we will now make a few remarks on this case. As noted in Proposition \ref{prop:Regularity-conditions-integer-nu}, the regularity conditions $(ii)$ guarantees the smoothness of $a_{p;\ell,\bmm}(\tau)$ in the interval $\tau \in [-1, 1]$. This is clear since $P^{p}_{\nu}(\tau) (1-\tau^2)^{\tfrac{p}{2}}$ is a polynomial of degree $\nu+p$ for $\nu \in \mathbb{Z}^{+}$. It is also useful to note that given initial data that satisfy the regularity conditions $(ii)$, the solution $a_{p;\ell,\bmm}(\tau)$ can be written as
\begin{equation}
    a_{p;\ell,\bmm}(\tau) = C_{p;\ell,\bmm} P^{p}_{\nu}(\tau) (1-\tau^2)^{\tfrac{p}{2}},~~~{\text {for}}~~~\nu\ge p.
    \label{Solution-with-regularity-conditions}
\end{equation}
 Substituting Eq.~\eqref{Associated-legendre-polynomial-integer-nu-expanded} in Eq. \eqref{Solution-with-regularity-conditions}, we get
\begin{equation}
\label{solution-general}
    a_{p;\ell,\bmm}(\tau) = C_{p;\ell,\bmm} \sum_{i=p}^{\nu} \frac{(-1)^{i+p+\nu}}{2^{\nu} \nu!} \binom{p+\nu}{i} \nu^{(i)} \nu^{(p-i+\nu)} (1-\tau)^{\nu-i+p} (1+\tau)^{i}.
\end{equation}
Then, one can show that for $\nu\ge p$
\begin{equation*}
    a_{p;\ell,\bmm}(\pm 1) = \begin{cases} C_{p;\ell,\bmm} 2^{p} (\Gamma(1-p))^{-1} & \text{for } \tau =1, \\
        (-1)^{p+\nu} C_{p;\ell,\bmm}  2^{p} (\Gamma(1-p))^{-1} & \text{for } \tau =-1,
    \end{cases}
\end{equation*}
and in particular
\begin{equation}
\label{solution-nu-p}
    a_{p;\ell,\bmm}(\pm 1) = \begin{cases} C_{0;\ell,\bmm} & \text{for } p=0, \tau =1, \\
        (-1)^{\nu} C_{0;\ell,\bmm}  & \text{for } p=0, \tau =-1, \\
        0 & \text{for } p\neq0, \tau =\pm1.
    \end{cases}
\end{equation}
\begin{remark}
    The behaviour of the solution Eq. \eqref{solution-exp} for $\mu \equiv -p \in \mathbb{Z}^{-}$ can be deduced from the following connection formulas:
    \begin{eqnarray*}
        && P^{-p}_{\nu}(\tau) = (-1)^{p} \frac{\Gamma(\nu-p+1)}{\Gamma(\nu+p+1)} P^{p}_{\nu}(\tau), \\
        && Q^{-p}_{\nu}(\tau) = (-1)^{p} \frac{\Gamma(\nu-p+1)}{\Gamma(\nu+p+1)} Q^{p}_{\nu}(\tau)
    \end{eqnarray*}
   and, moreover, we have
    \begin{eqnarray*}
        && P^{p}_{-\nu-1}(\tau) = P^{p}_{\nu}(\tau), \\
        && \sin((\nu-p)\pi) Q^{p}_{-\nu-1}(\tau) = \sin((\nu+p)\pi) Q^{p}_{\nu}(\tau) - \pi \cos(\nu \pi) \cos(p \pi) P^{p}_{\nu}(\tau).
    \end{eqnarray*}
\end{remark}
\subsection{Solutions for $\nu \notin \mathbb{Z}$}
Observe that the solution discussed in the previous section is manifestly symmetric. In particular, the conditions in Proposition \ref{Regularity-conditions-integer-nu} ensure the $C^{\infty}$-regularity of the solution at both $\tau=1$ and $\tau=-1$. For $\nu \notin  \mathbb{Z}$, one can show that the conditions ensuring the $C^{0}$-regularity of the solution given by Eq.~\eqref{solution-exp} at $\tau=1$ are not equivalent to those ensuring its regularity at $\tau=-1$ and that requiring both conditions to be satisfied implies a trivial solution for $a_{p;\ell,\bmm}(\tau)$. Instead, consider the manifestly symmetric solution given by 
\begin{equation}
    a_{p;\ell,\bmm} (\tau) = \left( C_{p;\ell,\bmm} (P^{p}_{\nu}(\tau)+P^{p}_{\nu}(-\tau)) + D_{p;\ell,\bmm} (Q^{p}_{\nu}(\tau)+Q^{p}_{\nu}(-\tau)) \right) (1-\tau^2)^{\tfrac{p}{2}}.
    \label{Solution-non-integer-nu}
\end{equation}
Using the connection formulas between $P^{p}_{\nu}(\tau), Q^{p}_{\nu}(\tau), P^{p}_{\nu}(-\tau)$ and $Q^{p}_{\nu}(-\tau)$, namely
\begin{subequations}
    \begin{eqnarray}
        && P^{p}_{\nu}(\tau) = \cos((\nu-p)\pi) P^{p}_{\nu}(-\tau) - \frac{2}{\pi} \sin((\nu-p)\pi) Q^{p}_{\nu}(-\tau), \\
        && Q^{p}_{\nu}(\tau) = - \cos((\nu-p)\pi) Q^{p}_{\nu}(-\tau) - \frac{\pi}{2} \sin((\nu-p)\pi) P^{p}_{\nu}(-\tau),
    \end{eqnarray}
    \label{Connection-formulas-3}
\end{subequations}
we can write $a_{p;\ell,\bmm} (\tau)$ as 
    \begin{equation}
    \label{solution-not-int}
    \aligned
        a_{p;\ell,\bmm} (\tau) &= \left(C_{p;\ell,\bmm} (1+\cos((\nu-p)\pi)) - \frac{\pi}{2} D_{p;\ell,\bmm} \sin((\nu-p)\pi) \right) P^{p}_{\nu}(\pm \tau) (1-\tau^2)^{\tfrac{p}{2}} \\
        &+ \left( D_{p;\ell,\bmm}  (1-\cos((\nu-p)\pi)) - \frac{2}{\pi} C_{p;\ell,\bmm} \sin((\nu-p)\pi) \right) Q^{p}_{\nu}(\pm \tau) (1-\tau^2)^{\tfrac{p}{2}},
        \endaligned
    \end{equation}
where the $\pm$ signs are chosen such that $P^{p}_{\nu}(\pm \tau) (1-\tau^2)^{\tfrac{p}{2}}$ has well-defined limits $\tau=\pm 1$. In particular since $\lim_{\tau\to -1} P^{p}_{\nu}(\tau)$ diverges, we are setting up our solution to avoid this divergent behaviour.

Starting from the limits \cite{DLMF}
    \begin{subequations}
        \begin{eqnarray}
        && \lim_{\substack{\tau \to \pm 1}} \left( \frac{2}{1\mp\tau} \right)^{\tfrac{p}{2}} P^{p}_{\nu}( \pm \tau) = (-1)^{p} \frac{(\nu-p+1)_{(2p)}}{p!}, \label{limit-P-p-nu} \\
        && \lim_{\substack{\tau \to \pm 1}} \left( \frac{1\mp \tau}{2} \right)^{\tfrac{p}{2}} Q^{p}_{\nu}( \pm \tau) = \frac{1}{2} \cos(p \pi) \Gamma(p),  \qquad p \neq 0,
        \end{eqnarray}
    \end{subequations}
    and the expansion of $Q_{\nu}(\pm\tau)$ as $\tau \to \pm 1$
    \begin{equation*}
        Q_{\nu}( \pm \tau) = \frac{1}{2} \log \left( \frac{1}{1\mp\tau} \right) - \psi(1) - \psi(\nu+1) + O((1 \mp \tau) \log(1 \mp \tau)), \qquad \text{as } \tau \to \pm 1,
    \end{equation*}
    where $(x)_{(i)}:= x (x+1) \ldots (x+(i-1))$ and $\psi(z):= d (\log \Gamma(z))/dz$, one can show that
    \begin{eqnarray*}
        && \lim_{\substack{\tau \to \pm 1}}  P^{p}_{\nu}(\pm \tau) (1-\tau^2)^{\tfrac{p}{2}} \begin{cases}
                  =0   & \text{ for } p \neq 0, \\
                  =1 & \text{ for } p = 0,
            \end{cases} \\
        && \lim_{\substack{\tau \to \pm 1}}  Q^{p}_{\nu}(\pm \tau) (1-\tau^2)^{\tfrac{p}{2}} \begin{cases}
                  = 2^{p-1} \cos(p \pi) \Gamma(p)   & \text{ for } p \neq 0, \\
                  \sim \frac{1}{2} \log \left( \frac{2}{1 \mp \tau} \right) - \psi(1) - \psi(\nu+1) + O((1 \mp \tau) \log(1 \mp \tau)) & \text{ for } p = 0.
        \end{cases}
    \end{eqnarray*}
Finally, by taking the limits $\tau \to \pm 1$ of Eq.~\eqref{solution-not-int}, we find that $a_{p;\ell,\bmm}(\tau)$ has a well-defined limit at $\tau = \pm 1$ only if the coefficient of $\lim_{\substack{\tau \to \pm 1}} Q^{p}_{\nu}(\pm \tau) (1-\tau^2)^{\tfrac{p}{2}}$ vanishes for $p=0$. Therefore, we have the following result:
        \begin{proposition} \label{prop:regularity-conditions-nu-non-integer}
            Consider $\nu  \notin\mathbb{Z}$. Then
            \begin{enumerate} 
            \item  For $p=0$, the solution given by Eq. \eqref{Solution-non-integer-nu} has well-defined limits at $ \tau = \pm 1$ if and only if 
            \begin{equation*}
                C_{0;\ell,\bmm} = \frac{\pi}{2} D_{0;\ell,\bmm} \tan \left( \frac{\nu}{2} \pi \right)\ne 0
            \end{equation*}
            \label{regularity-conditions-nu-non-integer}
       but this implies
        \begin{equation*}
            a_{0,\ell,\bmm}(\tau) = 0.
        \end{equation*}
        \item For $p \neq 0$, the solution given by Eq.~\eqref{Solution-non-integer-nu} has a well-defined limit at $ \tau = \pm 1$  with
        \begin{equation}
        \label{p-non-zero}
          \hspace{-1cm}  a_{p;\ell,\bmm}(\pm1) = \left( D_{p;\ell,\bmm}  (1-\cos((\nu-p)\pi)) - \frac{2}{\pi} C_{p;\ell,\bmm} \sin((\nu-p)\pi) \right) 2^{p-1} \cos(p \pi) \Gamma(p).
        \end{equation}
        \end{enumerate}
         \end{proposition}

\begin{remark}
    \label{Non-integer-nu-regularity}
    \emph{As we will soon demonstrate, only $a_{0;\ell,\bmm}(\tau)$ is required for the computation of the asymptotic charges. Consequently, we do not pursue the conditions that ensure the $C^{\infty}$-regularity of $a_{p;\ell,\bmm}(\tau)$ when $p \neq 0$ and $\nu \notin \mathbb{Z}$ (for $p=0$ and $\nu \notin \mathbb{Z}$, our conditions only implies smoothness near $\tau = \pm 1$ because we get a trivial solution). In principle, one can employ the connection formulas provided in \cite{DLMF} to determine the asymptotic behaviour of the derivatives of $P^{p}_{\nu}(\tau)$ and $Q^{p}_{\nu}(\tau)$, or alternatively, utilise their explicit representations in terms of hypergeometric functions shown in Appendix \ref{Appendix:Solutions-Legendre-ODE}.}
\end{remark}
\section{Asymptotic charges for the spin-$0$ field}
\label{charges-sec}
In this section, we use the results of the previous section to calculate the spin-$0$ asymptotic charges at the critical sets $\mathcal{I}^{\pm}$. To do so, we will require an expansion for $\tilde{\phi}$ given Eq.~\eqref{Expansion-phi}. Recall that $\tilde{\phi} = \Theta^{n/2-1} \phi$, so $\tilde{\phi}$ can be written as
        \begin{equation}
            \tilde{\phi} (\tau,\rho,\theta^A)=  \Theta^{n/2-1}\sum_{p=0}^{\infty} \sum_{\ell=0}^{\infty} \sum_{\bmm} \frac{1}{p!} a_{p;\ell,\bmm}(\tau) Y_{\ell,\bmm} \rho^{p}.
            \label{expansion-tilde-phi}
        \end{equation}
    or using Eq. \eqref{Coordinate-transformation}
        \begin{equation}
            \tilde{\phi} (\tilde t,\tilde r, \theta^A)=  \frac{1}{\tilde{r}^{n/2-1}}\sum_{p=0}^{\infty} \sum_{\ell=0}^{\infty} \sum_{\bmm} \frac{1}{p!} a_{p;\ell,\bmm}\left( \frac{\tilde{t}}{\tilde{r}} \right) Y_{\ell,\bmm} \left( \frac{\tilde{r}}{\tilde{r}^2-\tilde{t}^2} \right)^{p}.
            \label{expansion-tilde-phi-spherical}
        \end{equation}
Now, consider the lowest order contributions to $p$ around $\rho=0$. Then, Eq.~\eqref{expansion-tilde-phi-spherical} can be written as 
\begin{equation*}
    \tilde{\phi}(\tilde{t},\tilde{r}, \tilde{\theta}^A) = \frac{1}{\tilde{r}^{n/2-1}} \sum_{\ell=0}^{\infty} \sum_{\bmm} a_{0;\ell,\bmm}\left(\frac{\tilde{t}}{\tilde{r}} \right) Y_{\ell,\bmm}+ O(\tilde{r}^{2-n})
\end{equation*}
and near $\mathscr{I}^{\pm}$ 
the above can be written as
\begin{equation}
    \tilde{\phi} (\tilde{r}, \tilde{\theta}^A)  = \frac{\tilde{\phi}^{(n/2-1)}|_{\mathcal{I}^{\pm}}(\tilde{\theta}^{\mathcal{A}})}{\tilde{r}^{n/2-1}} + O(\tilde{r}^{2-n}),
\end{equation}
where\footnote{For $n=4$ note that these conditions coincide with the ones of \cite{NguyenWest22}.}  $\tilde{\phi}^{'(n/2-1)}\equiv \bmpartial_{\tilde{t}} \tilde{\phi}^{(n/2-1)}=0$ and 
\begin{equation}
    \tilde{\phi}^{(n/2-1)}|_{\mathcal{I}^{\pm}}(\tilde{\theta}^{\mathcal{A}}) \equiv \sum_{\ell=0}^{\infty} \sum_{\bmm} a_{0;\ell,\bmm}(\pm 1) Y_{\ell,\bmm}.
    \label{tilde-phi-first-order-critical-sets}
\end{equation}
With the above, we are now ready to introduce the asymptotic charges associated with spin-$0$ fields. Following the discussion in  \cite{NguyenWest22} (see also \cite{HennTroe19, NguyenWest22-2}), one can show that in $n$-dimensional Minkowski spacetimes, the asymptotic expansion of $\tilde{\phi}$ is preserved by Poincar\'{e} transformations if $\bmpartial_{\tilde{t}} \tilde{\phi}^{(n/2-1)}=0$. Therefore, the asymptotic charges on a cut $\mathcal{C}$ of null infinity $\mathscr{I}^{\pm}$ in $n$-dimensional spacetime can be written as \footnote{Compare with the definition of the conserved charges in \cite{FuenHenn24}.}
\begin{equation*}
    \mathcal{Q}_{n}(\lambda, \mathcal{C}) = \oint_{\mathcal{C}} \varepsilon_{n-2} \lambda(\tilde{\theta}^{\mathcal{A}}) \tilde{\phi}^{(n/2-1)}, 
\end{equation*}
where $\varepsilon_{n-2}$ is the area element on $\mathcal{C}$ and $\lambda(\tilde{\theta}^{\mathcal{A}}) $ is an arbitrary function of the angles. Given Eq. \eqref{tilde-phi-first-order-critical-sets}, setting $\mathcal{C}=\mathcal{I}^{\pm}$ and choosing $\lambda(\tilde{\theta}^{\mathcal{A}}) = Y_{\ell',\bmm'}$, where $\bmm' = (\ell'_{1}, \ell'_{2}, \ldots, \ell'_{n-3})$, we get
\begin{eqnarray}
    && \mathcal{Q}_{n}(\lambda, \mathcal{I}^{\pm}) = \sum_{\ell=0}^{\infty} \sum_{\bmm} a_{0;\ell,\bmm}(\pm 1) \oint_{\mathcal{I}^{\pm}}  \varepsilon_{n-2}  Y_{\ell,\bmm}  Y_{\ell',\bmm'} \nonumber \\
    && \phantom{\mathcal{Q}_{n}(\lambda, \mathcal{I}^{\pm})} = \sum_{\ell=0}^{\infty} \sum_{\bmm} a_{0;\ell,\bmm}(\pm 1) \delta_{\ell \ell'} \delta_{\bmm \bmm'} \nonumber  \\ 
    && \phantom{\mathcal{Q}_{n}(\lambda, \mathcal{I}^{\pm})} = a_{0;\ell',\bmm'}(\pm 1). \label{charges}
\end{eqnarray}
Then, given the solution $a_{0;\ell,\bmm}(\tau)$ discussed in the previous section, we conclude with the following:
\begin{theorem}
Consider a spin-$0$ field $\phi$ on $n$-dimensional Minkowskian spacetime satisfying Eq.~\eqref{unphysical-spin-$0$-equation-expanded} and consider asymptotic solutions of the form Eq.~\eqref{Expansion-phi} with initial data satisfying the regularity condition of Propositions \ref{prop:Regularity-conditions-integer-nu} and \ref{prop:regularity-conditions-nu-non-integer} for the $\nu \in \mathbb{Z}$ and $\nu \notin \mathbb{Z}$ cases, respectively. Then, the spin-$0$ asymptotic charges given by Eq.~\eqref{charges} satisfy the following:
\begin{enumerate}
\item For even $n\geq4$ and $l\ge 0$, there are infinitely many conserved charges at ${\cal I}^\pm$ given by:
\begin{eqnarray*}
    \mathcal{Q}_{4}(Y_{\ell, \bmm}, \mathcal{I}^{\pm}) = \begin{cases} C_{0;\ell,\bmm} & \text{for } \tau =1, \\
        (-1)^{\ell} C_{0;\ell,\bmm}  & \text{for } \tau =-1.
    \end{cases}
\end{eqnarray*}
\item For odd $n>4$, $l=0$, there is at least one \emph{non-trivial} conserved charge at ${\cal I}^\pm$ given by
\begin{equation*}
    \mathcal{Q}_{n}(Y_{0, \bm0}, \mathcal{I}^{\pm}) = C_{0;0,\bm0}.
\end{equation*}
\item For odd $n>4$ and $l>0$, there exist no non-trivial charges with well-defined limits at $\mathcal{I}^{\pm}$.
\end{enumerate}
\end{theorem}

\begin{remark}
  \emph{We emphasise that our results are formal and tied to the ansatz given by Eq.~\eqref{Expansion-phi} and hence only initial data covered by Eq.~\eqref{initial-data} is considered.  Notice that in the even $n>4$ case, if the regularity conditions in Proposition \ref{prop:Regularity-conditions-integer-nu} are not imposed, the charges are not well-defined at $\mathcal{I}^{\pm}$. In other words, even within the class of initial data covered by Eq. \eqref{initial-data}, only fine-tuned initial data will lead to well-defined asymptotic charges.  In the odd $n>4$ case, the conditions for $C^{0}$-regularity of our solution near $\tau= \pm 1$ implies a trivial solution in the $p=0$ case. Consequently, there exists no non-trivial asymptotic charges with well-defined limits at $\tau = \pm 1$.}
\end{remark}

\begin{remark}
  \emph{A similar result for the spin-$0$ field was obtained for the Newman-Penrose constants in \cite{Gasperin-Pinto} in the $n=4$ case in the sense that only fine-tuned initial data lead to a finite expression for the \emph{classical Newman-Penrose constants}. However, also in \cite{Gasperin-Pinto}, it was observed that by considering a particular case of the so-called \emph{modified Newman-Penrose constants} ---see \cite{GajKeh22}, without imposing the analogous regularity condition, one gets a finite expression for these modified Newman-Penrose constants. Then, an interesting question is whether there exist analogous modified asymptotic charges for which one could get a regular result without imposing the regularity conditions. }
\end{remark}
\subsubsection*{Acknowledgements}
{\small The authors wish to thank Salvatore Vultaggio for valuable discussions on hypergeometric functions. 
MM would like to thank Kevin Nguyen, Félicien Comtat, Grigalius Taujanskas, and Juan A. Valiente Kroon for useful discussions while working on this project as well as David Fajman
for his hospitality. MM gratefully acknowledges support from the Erwin Schrödinger International Institute for Mathematics and Physics' (ESI) Junior Research Fellowship and the Centre for Mathematical Analysis, Geometry, and Dynamical Systems' (CAMGSD) Postdoctoral Fellowship. MM and FCM thank CAMGSD, IST-ID, projects UIDB/04459/2020 and UIDP/04459/2020. E. Gasper\'in holds an FCT (Portugal) investigator grant
2020.03845.CEECIND and benefited from the FCT Exploratory Research Project 2022.01390.PTD to fund academic visits where part of this research was conducted. FCM thanks CMAT, Univ. Minho, through FCT projects UIDB/00013/2020 and UIDP/00013/2020 and FEDER Funds COMPETE. The authors also benefited from the H2020-MSCA-2022-SE project EinsteinWaves, GA No.101131233. Some calculations in this project used the computer algebra system Mathematica with the package xAct \cite{xAct}.

\appendix
\setcounter{equation}{0}  
\renewcommand{\theequation}{\thesection.\arabic{equation}}

\section{Solutions to the Legendre's differential equation}
\label{Appendix:Solutions-Legendre-ODE}
Following the discussion in Chapter 3 in \cite{ErdBat06} (see also \cite{GradshteynRyzhik2014,DLMF}), consider the second order ordinary differential equation
\begin{equation}
    (1-z^2) u''(z) - 2 z u'(z) + \left( \nu (\nu+1) - \frac{\mu^2}{1-z^2} \right) u(z) =0,
    \label{differential-Equation-Appendix-A}
\end{equation}
where $u(z)$ is a function of the complex variable $z \in \mathbb{C}$, $u'(z) \equiv du(z)/dz$, $u''(z) \equiv d^{2}u(z)/dz^2$, and $\nu, \mu$ are some arbitrary complex constants. The solutions to Eq. \eqref{differential-Equation-Appendix-A} are multiple-valued functions in the $z$-plane. By making a branch cut along the real axis from $- \infty$ to $+1$, one obtains the linearly independent single-valued regular solutions of Eq. \eqref{differential-Equation-Appendix-A}, the associated Legendre functions of the first and second kind, denoted by $P^{\mu}_{\nu}(z)$ and $Q^{\mu}_{\nu}(z)$. In terms of the hypergeometric function ${}_{2}F_{1}$, $P^{\mu}_{\nu}(z)$ and $Q^{\mu}_{\nu}(z)$ can be written as  
\begin{subequations}
    \begin{eqnarray}
    && P^{\mu}_{\nu}(z) := \frac{1}{\Gamma(1-\mu)} \left(\frac{z+1}{z-1}\right)^{\frac{\mu}{2}} {}_{2}F_{1}\left( -\nu, \nu+1; 1-\mu; \tfrac{1}{2} (1-z) \right), \quad |1-z|<2, \quad \qquad \label{Assoc-Legendre-function-first-kind}\\
    && Q^{\mu}_{\nu}(z) := \frac{e^{i \pi \mu} \Gamma(\nu+\mu+1) \Gamma(\frac{1}{2})}{2^{\nu+1} \Gamma(\nu+ \frac{3}{2})} (z^2-1)^{\frac{\mu}{2}} z^{-\nu - \mu -1} \nonumber \\
    && \phantom{\mathcal{Q}^{\mu}_{\nu}(z) = \frac{e^{i \pi \mu}}{2^{\nu+1}}} \times {}_{2}F_{1}\left(\frac{\nu + \mu +2}{2},\frac{\nu + \mu +1}{2};\nu + \frac{3}{2}; \frac{1}{z^2} \right), \qquad |z|>1, 
\end{eqnarray}
\end{subequations}
where $\Gamma$ denotes the Gamma function. Different expressions of $P^{\mu}_{\nu}(z)$ and $Q^{\mu}_{\nu}(z)$ can be obtained by using transformation formulas of the hypergeometric functions --- See \cite{ErdBat06}. Note that Eq. \eqref{differential-Equation-Appendix-A} is invariant under the replacements $\mu \to - \mu$, $z \to - z$ $\nu \to - \nu - 1$. Therefore, $P^{\pm \mu}_{\nu}(\pm z), Q^{\pm \mu}_{\nu}(\pm z), P^{\pm \mu}_{-\nu -1}(\pm z)$ and $Q^{\pm \mu}_{-\nu -1}(\pm z)$ are also solutions to Eq. \eqref{differential-Equation-Appendix-A}. 

Now, consider the solutions to Eq. \eqref{differential-Equation-Appendix-A} with a real argument $\tau \in \mathbb{R} \in [-1,1]$. If $\mu=0$ and $\nu \in \mathbb{Z}^{+}$, the solutions to Eq. \eqref{differential-Equation-Appendix-A} are given by
\begin{subequations}
    \begin{eqnarray}
    && P_{\nu} (\tau) := \frac{(-1)^{\nu}}{2^\nu \nu!} \frac{d^{\nu} (1-\tau^2)^{\nu}}{d \tau^{\nu}}, \label{Definition-Legendre-polynomial} \\
    && Q_{\nu} (\tau) := \frac{1}{2} P_{\nu}(\tau) \log \left( \frac{1+\tau}{1-\tau} \right) - W_{\nu-1}(\tau), \label{Definition-Legendre-function-2}
\end{eqnarray}
\end{subequations}
where $P_{\nu}(\tau)$ is the Legendre function of the first kind, $Q_{\nu}(\tau)$ is the Legendre function of the second kind and $W_{\nu-1}$ is a polynomial of degree $\nu -1$ that satisfies the equation
\begin{equation*}
    (1-\tau^2) \frac{d^{2}W_{\nu-1}}{d \tau^2} - 2 \tau \frac{d W_{\nu-1}}{d \tau} + \nu (\nu+1) W_{\nu-1} = 2 \frac{d P_{\nu}}{d \tau}. 
\end{equation*}
In particular, it can be shown that 
\begin{equation*}
    W_{\nu-1}(\tau) = \sum_{i=0}^{\frac{\nu-1}{2}} \frac{(2\nu - 4 i -1)}{(\nu - i)(2 i +1)} P_{\nu - 2 i -1 }(\tau).
\end{equation*}
\begin{remark}
    It is clear from Eq. \eqref{Definition-Legendre-function-2} that for $\nu \in \mathbb{Z}^{+}$, $ Q_{\nu} (\tau)$ diverges logarithmically at $\tau=\pm 1$.
\end{remark}
For $\nu \in \mathbb{R}$, one can write $P_{\nu}(\tau)$ and $Q_{\nu}(\tau)$ as follows
\begin{subequations}
    \begin{eqnarray}
    && P_{\nu} (\tau) = {}_{2}F_{1}\left( 1+\nu,-\nu;1; \tfrac{1}{2} (1-\tau) \right), \label{Definition-Legendre-polynomial-noninteger-nu} \\
    && Q_{\nu} (\tau) = \frac{1}{2} P_{\nu}(\tau) \left( \log \left(\frac{1+\tau}{1-\tau}\right) - 2 \psi(1) - 2 \psi(\nu+1) \right) - \pi^{-1} \sin(\nu \pi)   \\
    && \phantom{Q_{\nu} (\tau) = \frac{1}{2} P_{\nu}(\tau) \left( \log - 2 \psi(1) - 2 \psi(\nu+1) \right)} \times \sum_{i=1}^{\infty} \frac{\Gamma(i-\nu) \Gamma(i+\nu+1) \sigma(i) (1-\tau)^i}{2^{i}(i!)^2}, \nonumber
\end{eqnarray}
\end{subequations}
where $\sigma(i) = \psi(i+1) + \psi(1)$ and $\psi(z)$ is defined as 
        \begin{equation*}
            \psi(z) = \frac{d \log \Gamma(z)}{d z}.
        \end{equation*}
For $\mu = p \in \mathbb{Z}$, the solutions $P^{p}_{\nu}(\tau)$ and $Q^{p}_{\nu}(\tau)$ can be written as 
\begin{subequations}
    \begin{eqnarray}
    && P^{p}_{\nu}(\tau) = (-1)^{p}(1-\tau^2)^{\tfrac{p}{2}} \frac{d^{p} P_{\nu}(\tau)}{d\tau^p} \equiv (-1)^{p}(1-\tau^2)^{\tfrac{p}{2}} P^{(p)}_{\nu}, \label{Definition-associated-legendre-polynomial} \\
    && Q^{p}_{\nu}(\tau) = (-1)^{p}(1-\tau^2)^{\tfrac{p}{2}} \frac{d^{p} Q_{\nu}(\tau)}{d\tau^p} \equiv (-1)^{p}(1-\tau^2)^{\tfrac{p}{2}} Q^{(p)}_{\nu}, \qquad -1 < \tau < 1, \label{Definition-associated-legendre-function-2}
\end{eqnarray}
\end{subequations}
In terms of ${}_{2}F_{1}$, we have
\begin{subequations}
    \begin{eqnarray}
    && P^{p}_{\nu}(\tau) = (-1)^{p} \frac{\Gamma(\nu+p+1)}{2^{p}\Gamma(\nu-p+1)} (1-\tau^2)^{\tfrac{p}{2}} {}_{2}F_{1}\left(\nu+p+1,p-\nu;p+1; \tfrac{1}{2} (1-\tau) \right), \qquad \quad \\
    && Q^{p}_{\nu}(\tau) = - \frac{1}{2} \pi \sin\left( \tfrac{1}{2} (p+\nu) \pi \right) \omega_{1}(p,\nu,\tau) + \frac{1}{2} \pi \cos \left( \tfrac{1}{2} (p+\nu)  \pi \right) \omega_{2}(p,\nu,\tau), 
\end{eqnarray}
\label{Solutions-nu-non-integer}
\end{subequations}
where
\begin{subequations}
    \begin{eqnarray*}
    && \omega_{1}(p,\nu,\tau) :=  \frac{2^{p}\Gamma\left(\frac{1}{2}(\nu+p+1)\right)}{\Gamma\left(\frac{1}{2}(\nu-p)+1\right)}\left(1-\tau^{2}\right)^{-\tfrac{p}{2}}{}_{2}F_{1}\left( -\tfrac{1}{2} (\nu+p), \tfrac{1}{2} (\nu-p+1);\tfrac{1}{2};\tau^{2}\right) , \\
    && \omega_{2}(p,\nu.\tau) := \frac{2^{p}\Gamma\left(\frac{1}{2}(\nu+p)+1\right)}{\Gamma\left(\frac{1}{2}(\nu-p+1)\right)}\tau\left(1-\tau^{2}\right)^{-\tfrac{p}{2}}{}_{2}F_{1}\left( -\tfrac{1}{2} (\nu+p-1),\tfrac{1}{2} (\nu-p)+1;\tfrac{3}{2};\tau^{2}\right). \qquad \quad
\end{eqnarray*}
\label{Definition-omega-12}
\end{subequations}
Finally, $Q^{p}_{\nu}(\tau)$ can be written in terms of $P^{p}_{\nu}(\tau)$ as
\begin{eqnarray*}
    && \frac{2 \pi}{\sin(\nu \pi) } Q^{p}_{\nu}(\tau) = \frac{\pi}{\sin(\nu \pi)} P^{p}_{\nu} (\tau) \left( \log \left( \frac{1+\tau}{1-\tau} \right)  - 2 \psi(1) - \psi(\nu+p+1) - \psi(\nu-p+1) \right) \\
    && \phantom{ \frac{2 \pi}{\sin(\nu \pi) } Q^{p}_{\nu}(\tau) } - e^{i p \pi} \left( \frac{1+\tau}{1-\tau} \right)^{\tfrac{p}{2}} \sum_{i=0}^{p-1} \frac{\Gamma(i-\nu) \Gamma(i+\nu+1) \Gamma(p-i) }{2^{i}i!} \cos(i \pi) (1-\tau)^{i} \nonumber \\
    && \phantom{ \frac{2 \pi}{\sin(\nu \pi) } Q^{p}_{\nu}(\tau) } - \left( \frac{1+\tau}{1-\tau} \right)^{\tfrac{p}{2}} \sum_{i=0}^{\infty} \frac{\Gamma(p+i-\nu) \Gamma(p+i+\nu+1)}{2^{p+i}i!(p+i)!} \sigma(i) (1-\tau)^{p+i} \nonumber \\
    && \phantom{ \frac{2 \pi}{\sin(\nu \pi) } Q^{p}_{\nu}(\tau) } - \frac{\Gamma(\nu+p+1)}{\Gamma(\nu-p+1)} \left( \frac{1-\tau}{1+\tau} \right)^{\tfrac{p}{2}} \sum_{i=0}^{\infty} \frac{\Gamma(i-\nu)\Gamma(i+\nu+1)}{2^{i}i! (p+i)!} \sigma(p+i) (1-\tau)^{i}. \nonumber
\end{eqnarray*}

\end{document}